\documentclass[letterpaper, 11pt, conference]{ieeeconf}      % Use this line for a4
\IEEEoverridecommandlockouts                              % This command is only
\usepackage{graphicx} % for pdf, bitmapped graphics files
\usepackage{epsfig} % for postscript graphics files
\usepackage{times} % assumes new font selection scheme installed
\usepackage{amsmath} % assumes amsmath package installed
\usepackage{amssymb}  % assumes amsmath package installed
\usepackage{cleveref}
\usepackage{subcaption,graphicx}
\usepackage{url}
\newtheorem{theorem}{Theorem}

\newtheorem{assumption}[theorem]{Assumption}
\usepackage{color}
\usepackage{cite}
\usepackage{mathrsfs}
\usepackage{multirow}
\usepackage{caption}
\usepackage{subcaption}
\usepackage{cite}
\usepackage{soul}

\newtheorem{definition}[theorem]{Definition}

{ % these settings define remarks and example as parts
\newtheorem{remark}[theorem]{Remark}

}
\usepackage{wrapfig}
\usepackage{float}
\newcommand{\R}{\mathbb{R}}

\newcommand{\bi}{\begin{itemize}}
\newcommand{\ei}{\end{itemize}}
\newcommand{\bd}{\begin{displaymath}}
\newcommand{\ed}{\end{displaymath}}
\newcommand{\be}{\begin{eqnarray*}}
\newcommand{\ee}{\end{eqnarray*}}

\def\R{\mathbb{R}}

\title{\LARGE \bf 
Data-Driven Optimal Control Using Perron-Frobenius Operator}
\author{Apurba Kumar Das, Bowen Huang, and Umesh Vaidya \\
\thanks{Financial support from the National Science Foundation grant CNS-1329915 and ECCS-1150405 is gratefully acknowledged.  The authors are with the Department of Electrical and Computer Engineering at Iowa State University. 
        {\tt\small ugvaidya@iastate.edu} }
        }

\begin{document}
\maketitle
\begin{abstract}
In this paper, we propose a data-driven approach for control of nonlinear dynamical systems. The proposed data-driven approach relies on transfer Koopman and Perron-Frobenius (P-F) operators for linear representation and control of such systems. Systematic model-based frameworks involving linear transfer P-F operator were proposed for almost everywhere stability analysis and control design of a nonlinear dynamical system in previous works \cite{VaidyaMehtaTAC, Vaidya_CLM,raghunathan2014optimal}. Lyapunov measure can be used as a tool to provide linear programming-based computational framework for stability analysis and almost everywhere stabilizing control design of a nonlinear system. In this paper, we show that those frameworks can be extended to a data-driven setting, where the finite dimensional approximation of linear transfer P-F operator and stabilizing feedback controller can be obtained from time-series data. We exploit the positivity and Markov property of these operators and their finite-dimensional approximation to provide {\it linear programming} based approach for designing an optimally stabilizing feedback controller.

\end{abstract}

\section{Introduction}
Stability analysis and stabilization of dynamical systems are two classical problems in control theory with applications ranging across various engineering discipline. Systematic tools exist for stability analysis and control design for linear systems, however, for nonlinear systems, this is still an active area of research. 
The introduction of linear transfer operator theoretic methods from dynamical system theory provides an opportunity to provide a systematic approach for the stability analysis and stabilization of nonlinear systems \cite{Lasota,Dellnitz00,Meic_model_reduction}. The transfer operator theoretic methods involving Perron-Frobenius (P-F) and Koopman operator provides for a linear representation of a nonlinear system by shifting the focus from the state space to the space of measures and functions. Linear nature of the transfer P-F operator was exploited to provide linear programming based systematic procedure for stability verification and optimal control design of nonlinear systems \cite{raghunathan2014optimal,das2017transfer}. In particular, Lyapunov measure and control Lyapunov measure were introduced for almost everywhere stability verification and design of stabilizing feedback controller for a nonlinear system. 

On the other hand, in this era of big data, there is new excitement towards developing data-driven methods for the analysis and control of complex dynamics \cite{DMD_schmitt,brunton2016discovering, kutz2016dynamic}. This excitement has lead to the renewed interest in the data-driven approximation of Koopman and P-F operators. These data-driven methods predominantly revolve around finite-dimensional approximation of Koopman operator, dual to transfer P-F operator \cite{rowley2009spectral,EDMD_williams,klus2015numerical,Umesh_NSDMD}. The Koopman operator is better suited for data-driven approximation compared to transfer P-F operator. Spectral analysis of Koopman operator and its finite dimensional approximation constructed from time-series data has been successfully applied to address analysis problems in several applications \cite{mezic_koopmanism,susuki2011nonlinear,surana_observer,MVCDC05}. 
There has also been attempt to extend their applicability for control design for nonlinear systems \cite{kaiser2017data,peitz2017koopman}. However, they do not exploit the real potential and linear nature of Koopman operator for control design. 
None of those provide a systematic linear programming-based approach for the design of controllers for nonlinear systems.
The main contribution of this paper is to show that systematic data-driven linear methods can be developed for optimal controller design of nonlinear system exploiting the true potential of the linear operator theoretic framework.

This main contribution towards developing systematic data-driven control design for a nonlinear system is made possible by utilizing not only the linearity but also positivity, Markov property, and duality between Koopman and P-F operators. In particular Naturally Structured Dynamic Mode Decomposition (NSDMD) algorithm provides a data-driven approximation of Koopman and P-F operators and preserves positivity and Markov properties of these operators \cite{Umesh_NSDMD}. The main contribution of our paper is to show that the NSDMD algorithm can be combined with systematic model-based transfer P-F operator approach to provide a data-driven linear programming-based method for optimal control of a nonlinear system.

\section{Preliminaries: Lyapunov measure and optimal stabilization }\label{section_lymeas} 
In this section we provide brief overview of the application of linear transfer P-F operator framework for almost everywhere stability analysis and optimal stabilization of nonlinear system using Lyapunov measure \cite{VaidyaMehtaTAC, Vaidya_CLM,raghunathan2014optimal}. 
%We first provide a brief overview of Lyapunov measure as a tool for almost everywhere stability analysis of nonlinear system. 

Consider the discrete-time dynamical systems of the form,
\begin{equation}
x_{n+1}=F(x_n),\label{system}
\end{equation}
where $F: X\rightarrow X$ is assumed to be continuous with $X\subset \mathbb{R}^q$, a compact set. We denote ${\cal B}(X)$ as the Borel-$\sigma$ algebra on $X$ and ${\cal M}(X)$ as the vector space of a real valued measure on ${\cal B}(X)$. The mapping $F$ is assumed to be nonsingular with respect to the Lebesgue measure $\ell$, i.e., $\ell(F^{-1}(B))=0$, for all sets $B\in {\cal B}(X)$, such that $\ell(B)=0$. In this paper, we are interested in data-driven optimal stabilization of an attractor set defined as follows:

\begin{definition}[Attractor set]  A set ${\cal A}\subset X$ is said to be forward invariant under $F$, if $F({\cal A})={\cal A}$. A
closed forward invariant set $\cal A$ is said to be an attractor set, if there exists a neighborhood $V\subset X$ of $\cal A$, such that $\omega(x)\subset {\cal A}$ for all $x\in V$, where $\omega(x)$ is the $\omega$ limit set of $x$.
\end{definition}

 \begin{remark}\label{remark_nbd}We will use the notation $U(\epsilon)$ to denote the $\epsilon>0$ neighborhood of the attractor set $\cal A$ and $m\in {\cal M}(X)$, a finite measure absolutely continuous with respect to Lebesgue.
 \end{remark}

\begin{definition}[a.e. stable with geometric decay]\label{stable_a.e.}
The attractor set ${\cal A}\subset X$ for a dynamical system (\ref{system}) is said to be almost everywhere (a.e.) stable with geometric decay with respect to some finite measure $m\in {\cal M}(X)$, if given any $\epsilon>0$, there exists $M(\epsilon)<\infty$ and $\beta<1$, such that
$m\{x\in {\cal A}^c : F^{n}(x)\in X\setminus U(\epsilon)\}<M(\epsilon) \beta^n$.
\end{definition}

The above set-theoretic notion of a.e. stability was introduced and verified by using the
linear transfer operator framework \cite{VaidyaMehtaTAC}. For the discrete time dynamical system (\ref{system}), the linear transfer Perron Frobenius (P-F) operator denoted by $\mathbb{P}_F: {\cal M}(X)\rightarrow {\cal M}(X)$ is given by,
\begin{equation}
[{\mathbb P}_F \mu](B)=\int_{X}\chi_{B}(F(x))d\mu(x)=\mu(F^{-1}(B)),
\end{equation}
where $\chi_B(x)$ is the indicator function supported on the set
$B\in {\cal B}(X)$ and $F^{-1}(B)$ is the inverse image of set $B$ \cite{Lasota}. We define a sub-stochastic operator as a
restriction of the P-F operator on the complement of the attractor
set as follows:
\begin{equation}
 [{\mathbb P}^1_F\mu](B):=\int_{{\cal A}^c}\chi_{B}(F(x))d\mu(x)\label{rest_PF},
\end{equation}
for any set $B\in {\cal B}({\cal A}^c)$ and $\mu\in {\cal M}({\cal A}^c)$.
The condition for the a.e. stability of an attractor set
$\cal A$ with respect to some finite measure $m$ is defined in terms of
the existence of the {\it Lyapunov measure} $\bar \mu$, defined as follows \cite{VaidyaMehtaTAC}.
\begin{definition}[Lyapunov measure]\label{def_Lya_measure}
The Lyapunov measure is defined as any non-negative measure $\bar
\mu$, finite outside
$U(\epsilon)$ (see Remark \ref{remark_nbd}),  and satisfies the following
inequality, $[{\mathbb P}_F^1 \bar \mu ](B)<\gamma^{-1} \bar \mu(B)$, for some $\gamma \geq1$ and all sets $B\in {\cal B}( X\setminus
U(\epsilon))$, such that $m(B)>0$.
\end{definition}

The following theorem provides the
condition for a.e. stability with geometric decay \cite{Vaidya_converse}.
\begin{theorem} \label{theorem_converse} An attractor set $\cal A$
for the dynamical system (\ref{system}) is a.e. stable with
geometric decay with respect to finite measure $m$, if and only if for all $\epsilon>0$ there exists
a non-negative measure $\bar \mu$, which is finite on
${\cal B}(X\setminus U(\epsilon))$  and satisfies
%and a $\gamma>1$  such that following equality is satisfied
\begin{equation}
\gamma[ {\mathbb P}_F^1\bar \mu](B)-\bar \mu(B)=-m(B)\label{LME},
\end{equation} 
for all measurable sets $B\subset X\setminus U(\epsilon)$ and for some $\gamma>1$ and where  the geometric decay rate is given by $\beta\leq \frac{1}{\gamma}<1$
\end{theorem}
\begin{proof}
We refer readers to Theorem \ref{theorem_converse} from \cite{Vaidya_converse} for the proof.

\end{proof}

\subsection{Lyapunov measure for stabilization}

We consider the stabilization of dynamical systems of the form $x_{n+1}=T(x_n,u_n)$,
where $x_n\in X\subset \mathbb{R}^q$ and $u_n \in U\subset
\mathbb{R}^d$ are respectively the states and the control inputs. Both $X$ and $U$ are assumed  compact. The
objective is to design a feedback controller $u_n=K(x_n)$, to
stabilize the attractor set $\cal A$. The stabilization problem is
solved using the Lyapunov measure by extending the P-F operator
formalism to the control dynamical system \cite{Vaidya_CLM}. We define
the feedback control mapping $C: X\rightarrow Y:= X\times U$ as
$C(x)=(x,K(x))$.
We denote ${\cal B}(Y)$ as the Borel-$\sigma$ algebra on $Y$ and ${\cal M}(Y)$ as the vector space of real valued measures on ${\cal B}(Y)$. For any $\mu \in {\cal M}(X)$, the control mapping $C$  can be used to define a measure $\theta\in {\cal M}(Y)$, as follows:
\begin{eqnarray}
&\theta(D):=[\mathbb P_C \mu](D)=\mu (C^{-1}(D))\nonumber\\
&[\mathbb P_{C^{-1}} \theta](B):=\mu(B)=\theta (C(B))\label{ess},
\end{eqnarray}
for all sets $D\in {\cal B}(Y)$ and $B\in {\cal B}(X)$. Since $C$ is an injective function with $\theta$ satisfying (\ref{ess}),  it follows from the theorem on disintegration of measure \cite{disintegration} (Theorem 5.8).  There exists a unique disintegration $\theta_x$ of the measure $\theta$  for $\mu$ and almost all $x\in X$, such that $
\int_Y f(y)d\theta(y)=\int_X\int_{C(x)} f(y)d\theta_x(y)d\mu(x)$,
for any Borel-measurable function $f:Y\to \mathbb R$.

%In particular, for $f(y)=\chi_D(y)$, the indicator function for the set $D$, we obtain
%$\theta(D)=\int_X \int_{C(x)}\chi_D(y) d\theta_x(y) d\mu(x)=[\mathbb{P}_C\mu](D).$
%Using the definition of the feedback controller
%mapping $C$, we write the feedback control system as $x_{n+1}=T(x_n, K(x_n))=T\circ C (x_n)$.
%The system mapping $T: Y\rightarrow X$ can be associated with P-F operators ${\mathbb
%P}_T: {\cal M}(Y)\rightarrow {\cal M}(X)$ as $
%[{\mathbb P}_T \theta](B)=\int_{Y}  \chi_B(T(y))d\theta(y)$.

This disintegration of $\theta$ measure allows us to write the P-F operator for the composition $T\circ C: X\rightarrow X$
as a product of ${\mathbb P}_T$ and ${\mathbb P}_C$ as follows:
\begin{eqnarray*}
&&[{\mathbb P}_{T\circ C} \mu](B)=\int_Y \chi_B(T(y))d [\mathbb{P}_C \mu](y)\nonumber\\&=&[\mathbb{P}_T \mathbb{P}_C \mu](B)=\int_X \int_{C(x)} \chi_{B}(T(y))d\theta_x(y)d\mu(x)\label{PTC}.
\end{eqnarray*}

The P-F operators $\mathbb{P}_T$ and $\mathbb{P_C}$ are used to define their restriction, $\mathbb{P}_T^1:{\cal M}({\cal A}^c\times U)\to {\cal M}({\cal A}^c)$, and $\mathbb{P}_C^1:{\cal M}({\cal A}^c)\to {\cal M}({\cal A}^c\times U)$ to the complement of the attractor set respectively, in a way similar to  Eq. (\ref{rest_PF}).
The control Lyapunov measure is defined as any non-negative
measure $\bar \mu\in {\cal M}({\cal A}^c)$,  finite on ${\cal
B}(X\setminus U(\epsilon))$, such that there exists a control mapping $C$ that satisfies following control Lyapunov measure equation
\begin{eqnarray}
\gamma [{\mathbb P}^1_T  {\mathbb P}^1_C \bar \mu ](B)- \bar \mu(B)=-m(B),\label{control_measure}
\end{eqnarray}
for every set $B\in {\cal B}( X\setminus U(\epsilon))$ and $\gamma\geq 1$.
Stabilization of the  attractor set is posed
as a co-design problem of jointly obtaining the control Lyapunov
measure $\bar \mu$ and the control P-F operator ${\mathbb P}_C$ \cite{Vaidya_CLM}.

%in particular disintegration of measure $\theta$, i.e., $\theta_x$. The disintegration measure $\theta_x$, which lives on the fiber of $C(x)$, in general, will not be absolutely continuous with respect to Lebesgue. For the deterministic control map, $K(x)$, the conditional measure, $\theta_x(u)=\delta(u-K(x))$, the Dirac delta measure. However, for the purpose of computation, we relax this  condition.
%The purpose of this paper and the following sections are to extend the Lyapunov measure-based framework for the optimal stabilization of nonlinear systems. One of the key highlights of this paper is the {\it deterministic} finite optimal stabilizing control is obtained as the solution for a finite linear program. 

In the following section we explain how the stabilization framework using Lyapunov measure can be extended to optimization stabilization using Lyapunov measure.

\subsection{Optimal stabilization}\label{optimal_control}

The basic idea behind the optimal stabilization is to augment the  control Lyapunov measure equation (\ref{control_measure}) with a cost function so that the attractor set $\cal A$ is stabilized while minimizing a certain cost. 

We consider the following cost function.
\begin{equation} {\cal
C}_C(B)=\int_{B}\sum_{n=0}^{\infty} \gamma^n G\circ C(x_n)
dm(x),\;\;\;\;\label{cost_sum}
\end{equation}
where $x_0=x$, the cost function $G: Y\to {\mathbb R}$ is assumed a continuous non-negative real-valued function,  such that $G({\cal A},0)=0$, $x_{n+1}=T\circ C(x_n)$, and $0<\gamma<\frac{1}{\beta}$. 
%Note, that in the cost function (\ref{cost_sum}), $\gamma$ is allowed  greater than one and this is one of the main departures from the conventional optimal control problem, where $\gamma\leq 1$. However, 
Under the assumption that the controller mapping $C$ renders the attractor set a.e. stable with a geometric decay rate, $\beta<\frac{1}{\gamma}$, the cost function (\ref{cost_sum}) is finite.
 %\begin{remark}
In the following we will use the notion of the scalar product between continuous function $h\in {\cal C}^0(X)$ and measure $\mu\in {\cal M}(X)$ as $\left<h,\mu\right>_X :=\int_X h(x)d\mu(x)$ \cite{Lasota}.
% In this notation  the weak convergence of measures has a simple form. Namely , $\{\mu_n\}$ converges to $\mu$ weakly if $\lim_{n\to \infty}\left< h,\mu_n\right>=\left<h,\mu\right>,\;\;\;\;\;\forall h\in C^0(X)$.
 %\end{remark}
 The following theorem proves  the cost of stabilization of the
 set $\cal A$ as given in Eq. (\ref{cost_sum}) can be expressed using the control Lyapunov measure equation.
\begin{theorem}\label{theorem_ocp} Let the controller mapping $C(x)=(x,K(x))$, be such that the attractor set $\cal A$ for the feedback control system $T\circ C:X\to X$ is a.e. stable with geometric decay rate $\beta<1$. Then, the cost function (\ref{cost_sum}) is well defined for $\gamma<\frac{1}{\beta}$ and, furthermore, the cost of stabilization of the attractor set $\cal A$ with respect to Lebesgue almost every initial condition starting from set $B\in {\cal B}(X_1)$ can be expressed as follows:
\begin{eqnarray}
&{\cal C}_C(B)= \int_{B}\sum_{n=0}^{\infty}\gamma^n G\circ C(x_n)
dm(x)\nonumber\\&=\int_{{\cal A}^c\times U}G(y)d[\mathbb{P}_C^1 \bar \mu_B](y)=
\left<G,{\mathbb P}_C^1 \bar \mu_B\right>_{{\cal A}^c\times
U},\label{inequality}
\end{eqnarray}
where $x_0=x$ and $\bar \mu_B$ is the solution of the following  control
Lyapunov measure equation,
\begin{eqnarray}
\gamma \mathbb{P}_T^1\cdot \mathbb{P}_C^1\bar \mu_B(D)-\bar
\mu_B(D)=-m_B(D), \label{control_Ly}
\end{eqnarray}
for all $D\in {\cal B}(X_1)$ and  where $m_B(\cdot):=m(B\cap \cdot)$ is a finite measure
supported on the set $B\in {\cal B}(X_1)$.
\end{theorem}
\begin{proof}
Refer to \cite{arvind_ocp_online} (Theorem \ref{theorem_ocp}) for the proof.
\end{proof}
 By appropriately selecting the measure on the right-hand side of the control Lyapunov measure equation (\ref{control_Ly}) (i.e., $m_B$),
stabilization of the attractor set with respect to
a.e. initial conditions starting from a particular set can
be studied.
The minimum cost of stabilization is defined as the minimum over
all a.e. stabilizing controller mappings $C$ with a geometric decay as follows:
\begin{equation}
{\cal C}^{*}(B)=\min_{C}{\cal C}_C(B).
\end{equation}
Using (\ref{inequality}) and (\ref{control_Ly}) the infinite dimensional linear program for optimal stabilization can be written as follows. We first define the projection map,
$P_1: {\cal A}^c\times U\rightarrow {\cal A}^c$
as:
$P_1(x,u)=x,$
and denote the P-F operator corresponding to $P_1$ as
$\mathbb{P}_{P_1}:  {\cal M}({\cal A}^c\times U)\rightarrow
{\cal M}({\cal A}^c)$, which can be written as
$[{\mathbb P}^1_{P_1} \theta](D)=\int_{{\cal A}^c\times U}\chi_D(P_1(y))d\theta(y)=\int_{D\times U}d\theta(y)=\mu(D)$.
 Using this definition of projection mapping $P_1$ and the
corresponding P-F operator, we can write the linear program for the optimal stabilization of set $B$
with  unknown variable
$\theta$ as follows:
\begin{eqnarray}
&\min\limits_{\theta\geq 0} \left<G, \theta\right>_{{\cal A}^c \times U},\nonumber\\
&\mbox{s.t. } \gamma [{\mathbb P}^1_T \theta](D)-[{\mathbb P}^1_{P_1}
\theta](D)=-m_B(D)\label{linear_program},
\end{eqnarray}
for $D\in {\cal B}(X_1)$.
%\begin{remark}\label{remark_discount}
%Observe  the geometric decay parameter satisfies $\gamma > 1$.
%This is in contrast to most optimization problems studied in
%the context of Markov-controlled processes, such as in Lasserre and
%Hern\'{a}ndez-Lerma \cite{LassHernBook}. Average cost and
%discounted cost optimality problems are considered in
%\cite{LassHernBook, viscosity_solnHJB_book}. The additional flexibility provided by $\gamma>1$ guarantees the controller obtained from the finite dimensional approximation of the infinite dimensional program (\ref{linear_program}) also stabilizes the attractor set for system (\ref{control_syst}). For a more detailed discussion on the role of $\gamma$ on the finite dimensional approximation, we refer  readers to the online version of the paper \cite{arvind_ocp_online}.
%\end{remark}

%%%%%%%%%%%%%%%%%%%%%%%%%%%%%%%%%%%%%%%%%%%%%%%%%%%%%%%%%%%%%%%%%%%%%%%%%%%%%%%%%%%%

\section{Data-Driven Approach for Optimal Stabilization
}\label{computation}

The computational framework relies on the finite dimensional approximation of the transfer P-F operator which is used in the approximation of infinite dimensional linear program for optimal stabilization. For the finite dimensional approximation of P-F operator from time-series data, we use Naturally Structured Dynamics Mode Decomposition (NSDMD) algorithm \cite{Umesh_NSDMD}. One of the distinguishing feature of this algorithm as compared to other algorithms available for finite dimensional approximation of Koopman and then P-F operator is it preserves two natural properties of these transfer operators namely positivity and Markov property. These properties are essential in the formulation of optimal stabilization problem as a linear program. In fact, in the absence of these properties, the optimal control problem for a nonlinear system using transfer operator framework cannot be formulated as a linear program. In the following, we briefly describe the NSDMD algorithm and then present the finite dimensional approximation of the linear program. 

\subsection{Naturally Structured Dynamic Mode Decomposition}\label{section_nsdmd}
The  Koopman operator corresponding to dynamical system (\ref{system}) is defined as
\[[\mathbb{U}h](x)=h(F(x)),
\] where $h\in {\cal C}^0(X)$. The Koopman and P-F operators are dual to each other and the duality is expressed as follows \footnote{With some abuse of notation we are using the same notation to define the P-F operator acting on the space of functions and measures}
\begin{eqnarray}\left<\mathbb{U}h,g\right>&=&\int_X[\mathbb{U}h](x)g(x)dx\nonumber\\
&=&\int_X h(x)[\mathbb{P}g](x)dx= \left<h,\mathbb{P}g\right>,
\end{eqnarray}
where $h\in{\cal L}_\infty(X)$ and $g\in {\cal L}_1(X)$ and the P-F operator on the space of densities are defined as follows: 
\begin{eqnarray}
[\mathbb{P}g](x)=g(F^{-1}(x))|\frac{dF^{-1}(x)}{dx}| \nonumber
\end{eqnarray}
%$[\mathbb{P}g](x)=g(F^{-1}(x))|\frac{dF^{-1}(x)}{dx}|$. 
Furthermore, these two operators also satisfy positivity property i.e., for any $h\geq 0$ and $g\geq 0$, we have $\mathbb{U}h\geq 0$ and $\mathbb{P}g\geq 0$. Another important property the P-F operator satisfies is the Markov property 
\[\int_X [\mathbb{P}g](x)d\mu(x)=\int_X g(x)d\mu(x),\]
where $\mathbb{P}: {\cal L}_1(X,\mu)\to {\cal L}_1(X,\mu) $ and $\mu$ is not necessarily invariant probability measure. 
The NSDMD algorithm approximate Koopman operator while preserving the positivity property. Furthermore, the duality between the Koopman and P-F operator combined with the Markov property of the P-F operator is exploited to provide data-driven approximation of P-F operator from the Koopman operator.  Hence  the NSDMD algorithm can be viewed as Extended Dynamic Mode Decomposition (EDMD) with added constraints to ensure positivity and Markov property. For the finite dimensional approximation, let $X=[x_1,\ldots,x_L]$ be the time-series data and \[{\bf \Psi}(x)=[\psi_1(x),\ldots, \psi_K(x)]^\top\]
as the choice of dictionary functions.
\begin{assumption}\label{assume}We assume that $\psi_j(x)\geq 0$ for $j=1,\ldots, K$ and define 
\begin{eqnarray}[\Lambda]_{ij}=\int_X \psi_i(x)\psi_j(x)dx.\label{lambda}
\\\nonumber
\end{eqnarray}
\end{assumption}

\begin{remark}\label{remark_grb}
In the simulation section we assume the dictionary functions to be Gaussian radial basis function for ensuring positivity of dictionary functions. The matrix $\Lambda$ in Eq. (\ref{lambda}) can be computed explicitly. 
\end{remark}
Under Assumption \ref{assume}, the finite dimensional approximation of Koopman operator ${K}\in \mathbb{R}^{K\times K}$,  and P-F operator $P\in \mathbb{R}^{K\times K}$, can be formulated as following optimization problem

\begin{eqnarray}\label{optimization_problem}
&\min\limits_{ K} \parallel {\bf G}{K}-{\bf A}\parallel_F\\\nonumber
&\text{s.t.} \;{K}_{ij} \geq 0,\;\;\;{\rm (Koopman \;positive\; constraints)}\\\nonumber
& [{\Lambda { K}\Lambda^{-1}}]_{ij}\geq 0,\;\;\;{\rm (P-F \;positive\; constraints)}\\\nonumber
& \Lambda{ K}\Lambda^{-1}{\bf 1} = {\bf 1},\;\;\;{\rm (P-F\; Markov\; constraints)}
\end{eqnarray}
where $\bf G$ and $\bf A$ are defined as follows:

\begin{eqnarray}\label{edmd2}
&&{\bf G}=\frac{1}{L}\sum_{m=1}^L \boldsymbol{\Psi}({x}_m)^\top \boldsymbol{\Psi}({x}_m)\nonumber\\
&&{\bf A}=\frac{1}{L}\sum_{m=1}^L \boldsymbol{\Psi}({x}_m)^\top \boldsymbol{\Psi}({y}_m).
\end{eqnarray}
and $\bf 1$ is the vector of all ones. The P-F operator $P$ is given by $P=\Lambda^{-1}K^\top \Lambda$.

\subsection{Finite Dimensional Approximation of Linear Program for Optimal Stabilization}

The finite dimensional approximation of the P-F operator can now be used in the finite dimensional approximation of the linear program in Eq. (\ref{linear_program}) for optimal stabilization. Towards this goal, we first discretize the control set $U$. The control input is quantized and  
assumed to take only finitely many control values from the
quantized set
${\cal U}_M = \{u^1,\hdots,u^a,\hdots,u^M\}$,
where $u^a \in \R^d$. For each fixed value of control input $u=u^a$, time-series data $\{x_1^a,\ldots, x_L^a\}$ for $a=1,\ldots, M$ is generated and the finite dimensional approximation of the P-F operator is constructed using the NSDMD algorithm outlined in section \ref{section_nsdmd}. We denote the P-F operator approximated for fixed value of control input $u=u^a$ as $P_a$. For the finite dimensional approximation of the infinite dimensional linear program we need to approximate the cost function $G$ and the measure $\theta$. Following Remark \ref{remark_grb} the centers for the Gaussian radial basis function are generated using K-mean clustering on data set generated from uncontrolled dynamical system. Let $x_\ell^*$ for $\ell=1,\ldots, K$ be the centers of the Gaussian radial basis functions. The finite dimensional approximation of the cost function is then expressed as  $G(x_\ell^*,u^a)$ for $\ell=1,\ldots, K$ and $a=1,\ldots, M$. Let $G_a=[G(x_1^*,u^a),\ldots, G(x_K^*,u^a)]^\top \in \mathbb{R}^K$ and $\theta_a\in \mathbb{R}^K$ be the finite dimensional approximation of measure $\theta$ on $X\times U$. The matrix representation of $\theta$ has $K$ rows and $M$ columns i.e., $\theta\in \mathbb{R}^{K\times M}$ \footnote{With some abuse of notations we are denoting both the infinite and finite dimensional representation of $\theta$ with same notation.}. The $(j,a)$ entry of $\theta$ is denoted by $\theta_{a}^j$ and we use the notation $\theta_a$ and $\theta^j$ for the $a^{th}$ column and $j^{th}$ row of $\theta$ respectively.  

Without loss of generality, we assume that the dictionary function $\psi_1(x)$ with center at $x_1^*$ is supported on the equilibrium point or the attractor set that we want to stabilize. Under this assumption, let $P_a^1\in \mathbb{R}^{(K-1)\times (K-1)}$ be the P-F matrix obtained from $P_a$  after deleting the first row and first column. Similarly, let $\bar \theta\in \mathbb{R}^{(K-1)\times M}$ be the matrix obtained from $\theta$ after deleting the first row.  
 $ G_a^1=[G(x_2^*,u^a),\ldots, G(x_K^*,u^a)]^\top\in \mathbb{R}^{K-1}$ is the vector obtained by deleting the first entry from vector $G_a$. The finite dimensional approximation of the infinite dimensional linear program (\ref{linear_program}) can then be written as follows:

\begin{eqnarray}
&\min\limits_{\bar \theta_a\geq 0} \;\;
\sum_{a=1}^M (G_a^1)^\top \bar \theta_a,\nonumber\\
&\mbox{s.t. } \gamma\sum_{a=1}^{M} (P_a^1)^\top\bar \theta_a
 - \sum_{a=1}^M\bar \theta_a = -m,\;\sum_a \bar \theta_a=\bf 1, 
\;\;\;\label{lp_finite}
\end{eqnarray}
\\
where $\bf 1$ is a vector of all ones and inequality $\bar \theta_a\geq 0$ is element-wise. The optimization problem (\ref{lp_finite}) is a finite dimensional linear program in terms of variable $\theta_a^1$. The solution to the optimization problem in general lead to a stochastic vector $\bar \theta^j$. The row vector $\bar\theta^j$ has a physical significance. In particular, $\bar \theta_a^j$ determines the probability of choosing the control action $a$ with state corresponding to the dictionary function $\psi_j(x)$. But, we are interested in determining deterministic control action i.e., 
\[\bar \theta^j_a=1 \;{\rm for\; exactly\; one} \;a\in \{1,\ldots,M\}. \]  
However, introducing this binary constraints on the entries of $\bar \theta^j$ in the optimization problem (\ref{lp_finite}) will lead to non-convex formulation which is difficult to solve.  Again, we know that deterministic control action can be obtained from stochastic $\bar \theta^j$ vector \cite{raghunathan2014optimal}. In particular, following choice of deterministic feedback control can be made from stochastic $\bar \theta^j$.
Let $\bar \theta_{a*}^j=\max \{\bar \theta_{1}^j,\ldots, \bar \theta_{M}^j\}$ i.e. $a*$ is the index corresponding to the maximum entry from the vector $\bar \theta^j$. Then the optimal deterministic feedback control is given by \[u^{a(j)}=u^{a*}. \]
At this point, the  optimal feedback control $k(x)$ is given by the following formula
\begin{eqnarray}
k(x)=\sum_{\ell=1}^K u^{a(j)}\psi_j(x).\nonumber
\end{eqnarray}

%\[k(x)=\sum_{\ell=1}^K u^{a(j)}\psi_j(x).\]

%\setlength{\belowcaptionskip}{-5pt}

\section{Simulation results}
In this section we provide results of the data-driven optimal stabilization algorithm applied into one-dimensional and two-dimensional continuous and discrete time nonlinear systems. Results are obtained using YALMIP with GUROBI solver coded in MATLAB.\\
\underline{\it Cubic Logistic Map}\\
Controlled equation for cubic logistic map is given as follows:
\begin{eqnarray}
x_{n+1} = \lambda x_n  - x_n^3 + u_n
\end{eqnarray}
where $x_n \in [-1.6,1.6]$ is the state, $u_n$ is the control input and we chose
parameter $\lambda = 2.3$. Let, control input space is quantized to $[-0.2:0.02:0.2]$. For the finite dimensional approximation of the P-F operator, we chose $200$ Gaussian radial basis functions  as dictionary function with $\sigma=0.008$. The cost function is assumed to be $x^2+u^2$.

 %We choose control to be discrete and in the range of $-0.2$ to $0.2$ with interval of $0.02$. We choose $200$ dictionary functions with $\sigma = 0.008$. Our goal is to find the control values which stabilizes the only equilibrium at origin while minimizing the cost function $x^2+u^2$. Fig. \ref{Fig:control_cubicLMap} contains the plot of optimal control values found by solving the Linear Program[\ref{}] which is expected to stabilize the equilibrium at origin.

%\begin{figure}[H]
%\begin{center}
%\includegraphics[scale = 0.48]{logs_cntrl}
%\caption{Control Measure for Cubic Logistic Map}
%\label{Fig:control_cubicLMap}
%\end{center}
%\end{figure}

%To verify stability, we take a certain neighborhood region near origin. For example, whenever the trajectories has reached anywhere between $-0.02$ and $0.02$ and stays in there, we assume stability has been achieved for an initial condition. Simulating trajectories for those initial conditions which were used during computation of the PF matrix, we had all the trajectories stabilized within the first $100$ iterations. Corresponding Lyapunov measure is in Fig. \ref{Fig:LyaMeas_cubicLMap}

%\begin{figure}[]
%\begin{center}
%\includegraphics[scale = 0.35]{CubicLog_imOL}
%\caption{{\textcolor{red}{Open Loop Invariant Measure for Cubic Logistic Map}}}
%\label{Fig:InvMeas_cubicLMap}
%\end{center}
%\end{figure}

\begin{figure}[]
\begin{center}
\includegraphics[scale = 0.29]{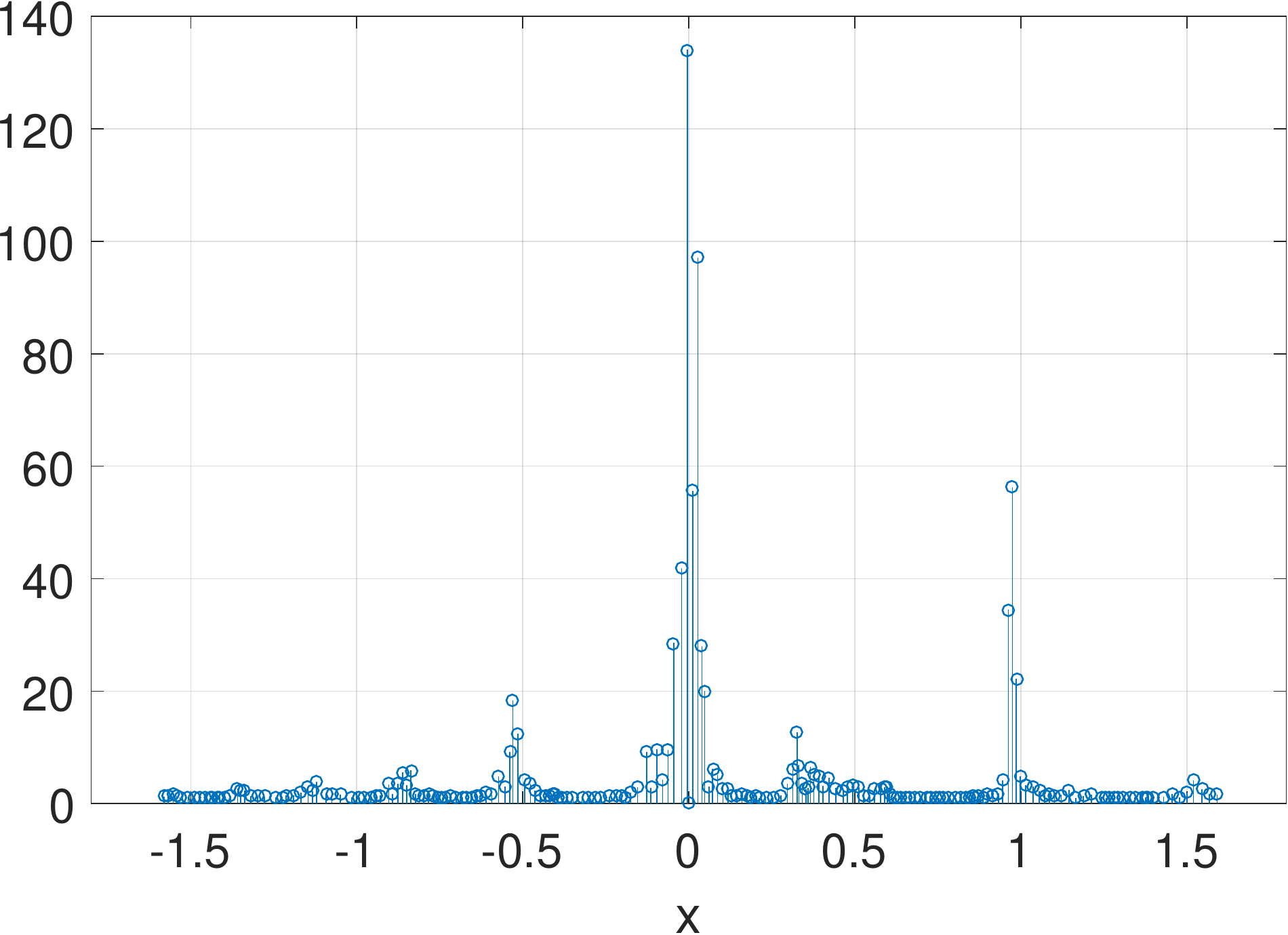}
\caption{Lyapunov Measure for Cubic Logistic Map}
\label{Fig2_cubiclogistic}
\end{center}
\end{figure}

In Fig. \ref{Fig2_cubiclogistic}, we provide the Lyapunov measure plot verifying the stability of the closed loop system. Fig. \ref{Fig1_cubiclogistic} shows two sample trajectories for the open loop and closed loop logistic maps. We observe that closed-loop trajectories are perfectly stabilized to the only equilibrium point at origin within few time steps.

\begin{figure}[]
\begin{center}
\includegraphics[scale = 0.36]{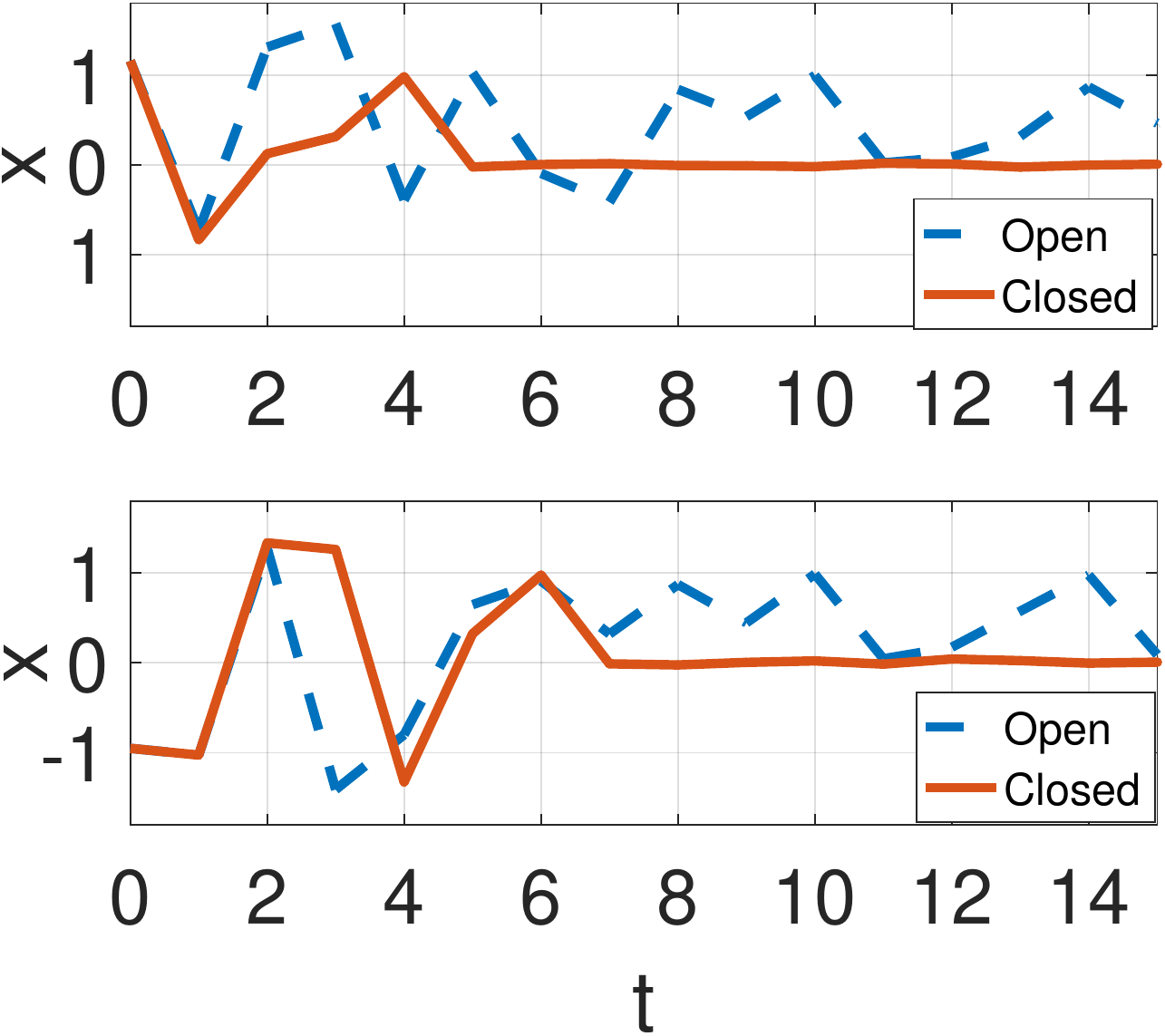}
\caption{Cubic Logistic Map: open loop and closed loop trajectories}
\label{Fig1_cubiclogistic}
\end{center}
\end{figure}

\underline{\it Duffing Oscillator}\\
The control of duffing oscillator is described by following equations
\begin{align}
\begin{split}
\dot{x}_1  &=   x_2 \\
\dot{x}_2  &=   (x_1-x_1^3) - 0.5 x_2 + u
\end{split}
\label{system_dynamics_duffingOsc}
\end{align}

\begin{figure}[]
\begin{center}
\includegraphics[scale = 0.4]{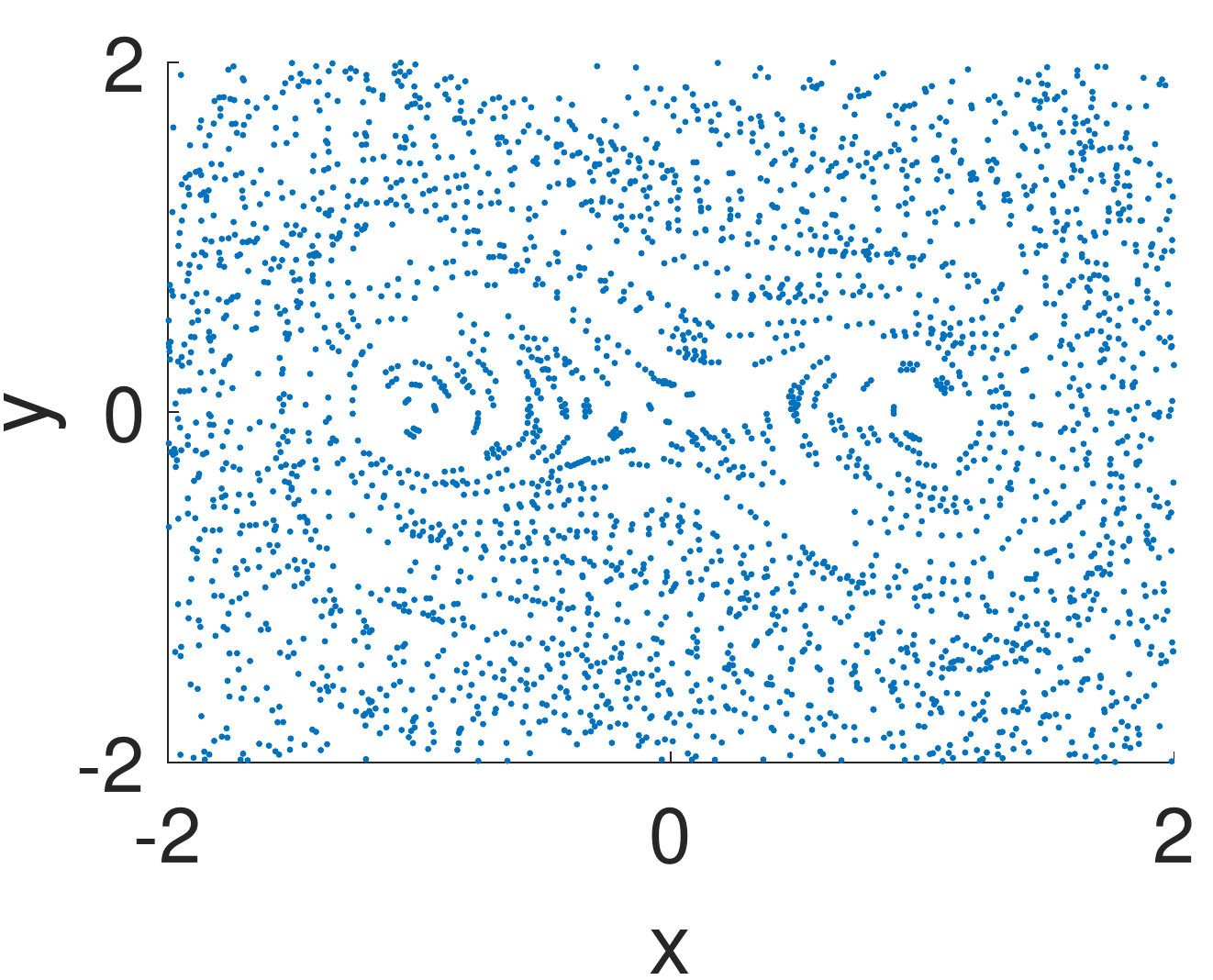}
\caption{Data for approximating transfer operator}
\label{Fig:duffing_data}
\end{center}
\end{figure}

The system has unstable equilibrium point at the origin and two stable equilibrium point at $(\pm 1,0)$. The objective is to stabilize the unstable equilibrium point at the origin. We consider the state space $X=[-2,2]\times [-2,2]$. For the finite dimensional approximation we use $100$ Gaussian radial basis function with $\sigma=0.2$. The centers for the radial basis functions are chosen using K-mean clustering algorithm applied to data set generated for open loop system and as shown in Fig. \ref{Fig:duffing_data}. The control input $u$ is quantized to ${\cal U} = [-4:0.5:4]$. In Figs. \ref{Fig1_duffing} and \ref{Fig2_duffing} we show the plots for the open loop and closed loop trajectories along with optimal cost and control inputs.

%$x_1$ denoting the position and $x_2$ denoting the velocity of a particle in a double well potential with unstable equilibrium $(0, 0)$ and stable equilibria $(\pm 1, 0)$. We study the phase space $X = [-2,2] \times [-2,2]$ for stabilizing the unstable equilibrium at origin. Here we used only $100$ dictionary functions with the choice of $\sigma = 0.2$ and  with control discretization ${\cal U} = \{-4,-3.5,-3,...,3,3.5,4\}$. This perfectly steers the trajectories towards the fixed point at origin for most of the initial conditions.

\begin{figure}[]
\begin{center}
 \centering
  \subcaptionbox{\label{duffos_traj420}}{\includegraphics[width=1.5in,height=1.2in]{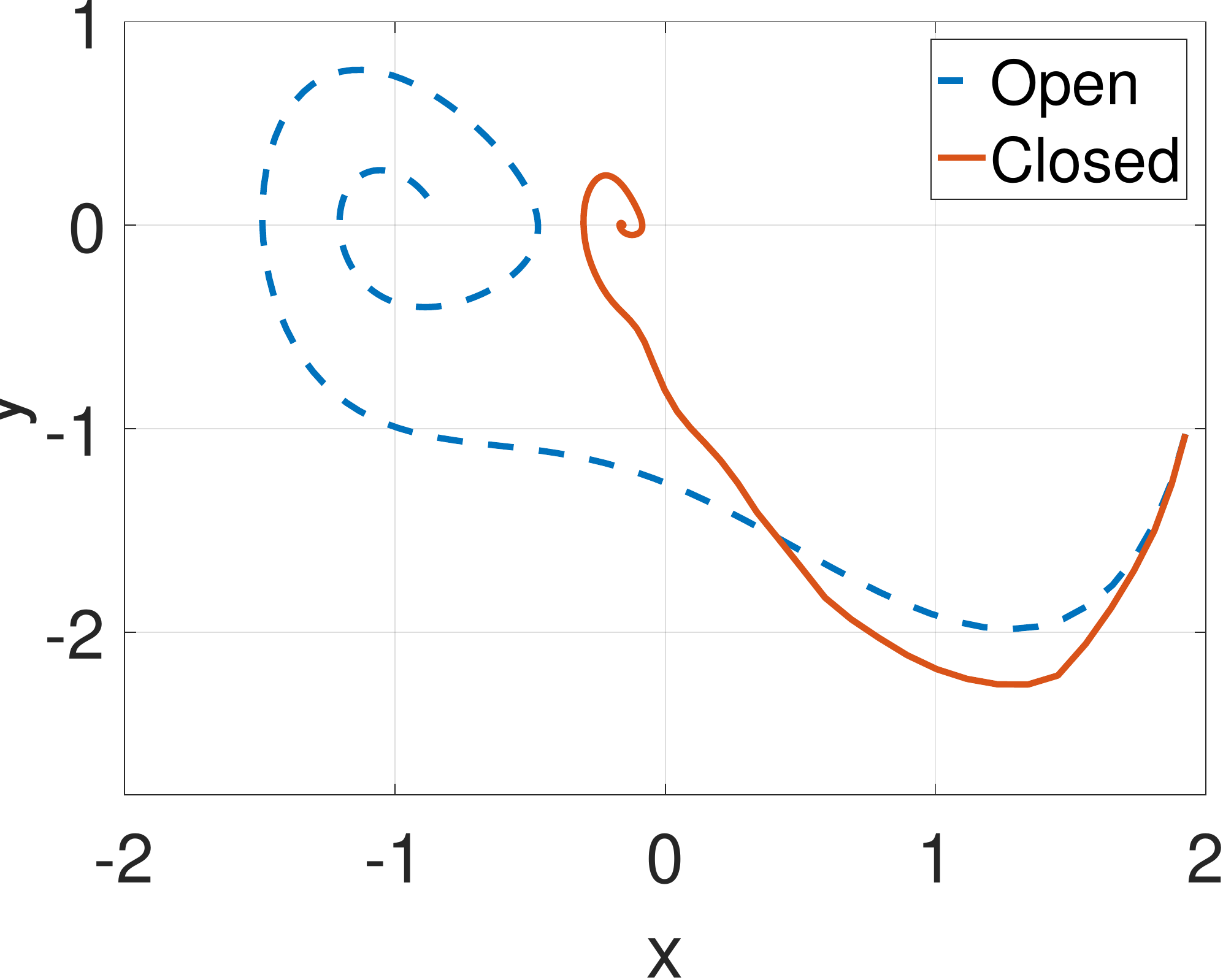}}\hspace{1em}%
  \subcaptionbox{\label{duffos_control_cost420}}{\includegraphics[width=1.5in,height=1.2in]{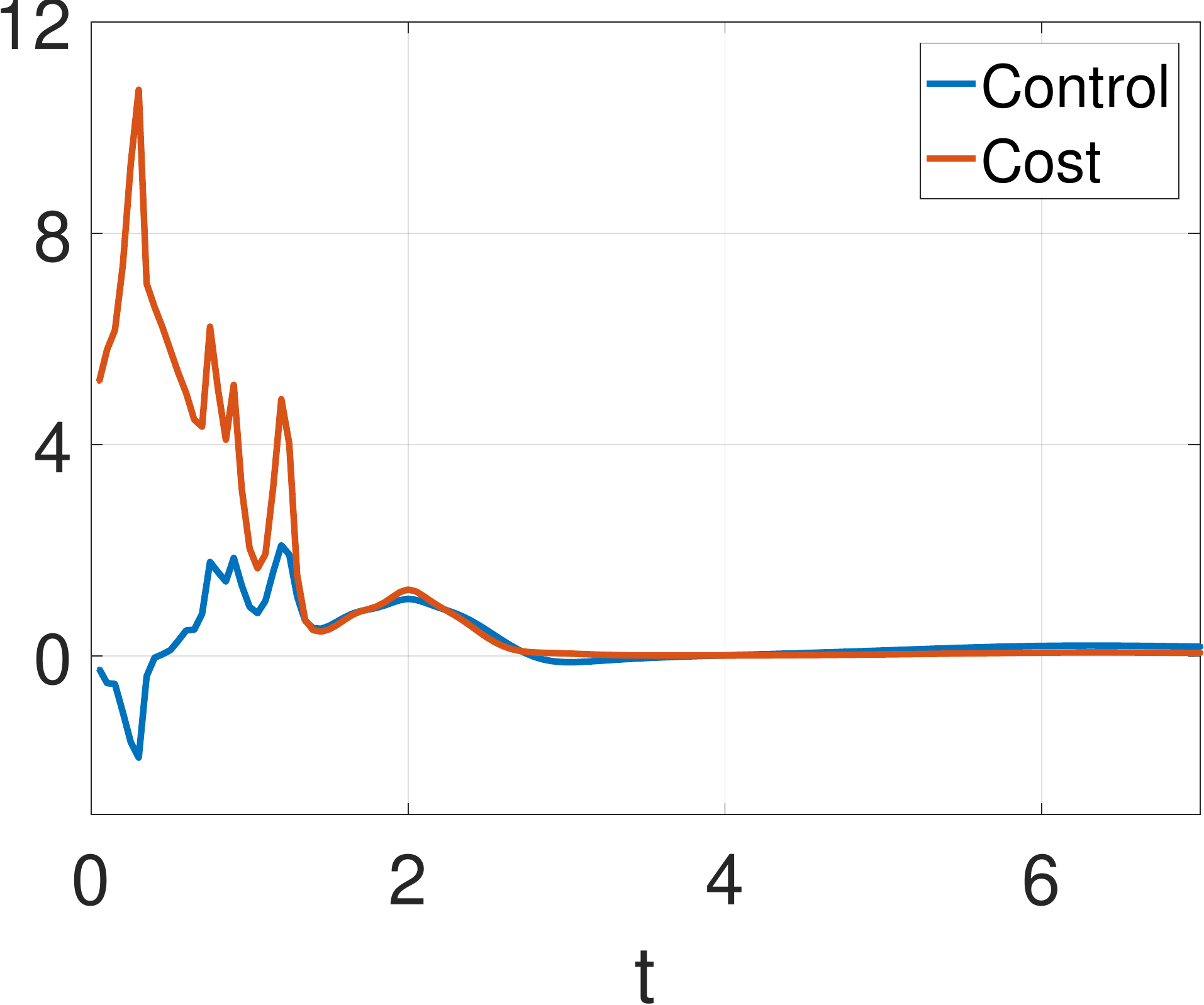}}
\caption{Duffing Oscillator: a) Open loop and closed loop trajectories; b) Optimal cost and control values.}
\label{Fig1_duffing}
\end{center}
\end{figure}

%% Legends for states trajectories
\begin{figure}[]
\begin{center}
 \centering
  \subcaptionbox{\label{duffos_traj712}}{\includegraphics[width=1.5in,height=1.2in]{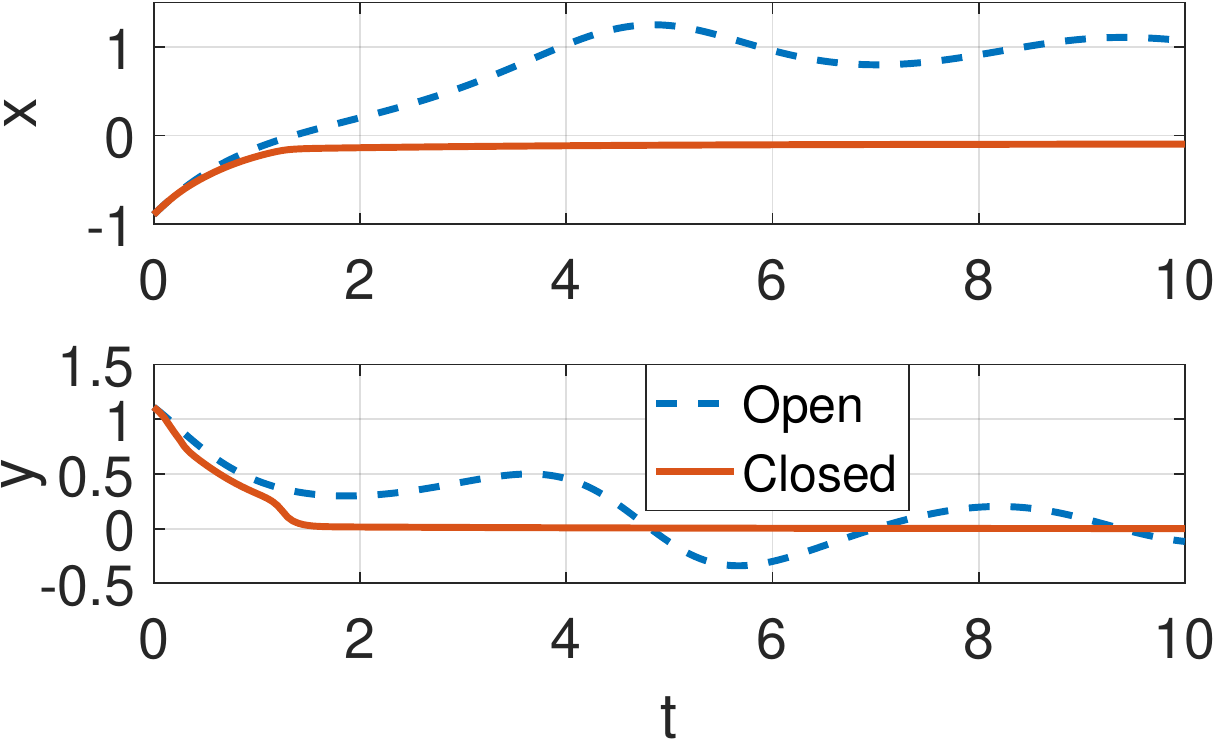}}\hspace{1em}%
  \subcaptionbox{\label{duffos_control_cost712}}{\includegraphics[width=1.5in,height=1.2in]{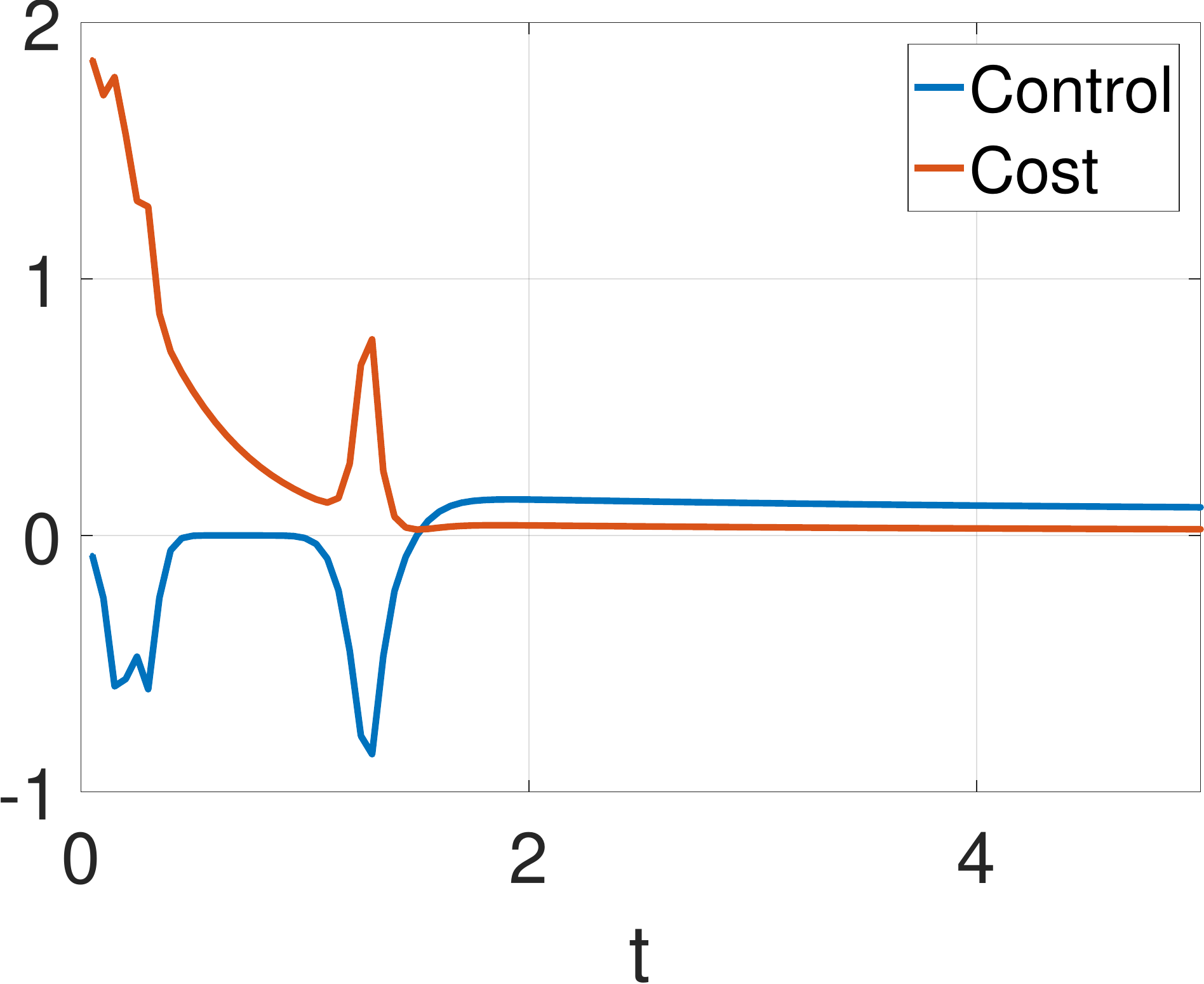}}
\caption{Duffing Oscillator: a) Open loop and closed loop trajectories; b) Optimal cost and control values.}
\label{Fig2_duffing}
\end{center}
\end{figure}

%We observe that the open loop trajectories are attracted to the stable equilibrium at $(-1,0)$ or $(1,0)$ whereas using the designed control the closed loop trajectories converges to the equilibrium at origin.

\underline{\it Basin Hopping in a Double Well}
\begin{eqnarray}
\dot{x_1}  &=& x_2 \nonumber \\
\dot{x_2}  &=&   -x_1^3 + ax_1^2 + x_1 - a + u\label{system_dynamics_well}
\end{eqnarray}
For parameter value of $a=0.5$, the system has three equilibrium points at $(\pm 1,0)$ and $(a,0)$. The equilibrium points at $(\pm 1,0)$ are stable and $(a,0)$ is unstable. The objective is to stabilize the unstable equilibrium point at $(a,0)$. Control quantization used for this example is ${\cal U}=[-2:0.2:2]$. For the finite dimension approximation, we construct $100$ Gaussian radial basis functions with $\sigma=0.22$.
Using the designed control, the intended unstable equilibrium was successfully stabilized for almost all the initial conditions. In figures \ref{Fig1_doublewell} and \ref{Fig2_doublewell}, we compare the open loop and close loop sample trajectories starting from two different initial conditions and corresponding optimal cost and control inputs.
\begin{figure}[]
\begin{center}
 \centering
  \subcaptionbox{\label{bhdwell_traj154}}{\includegraphics[width=1.57in,height=1.22in]{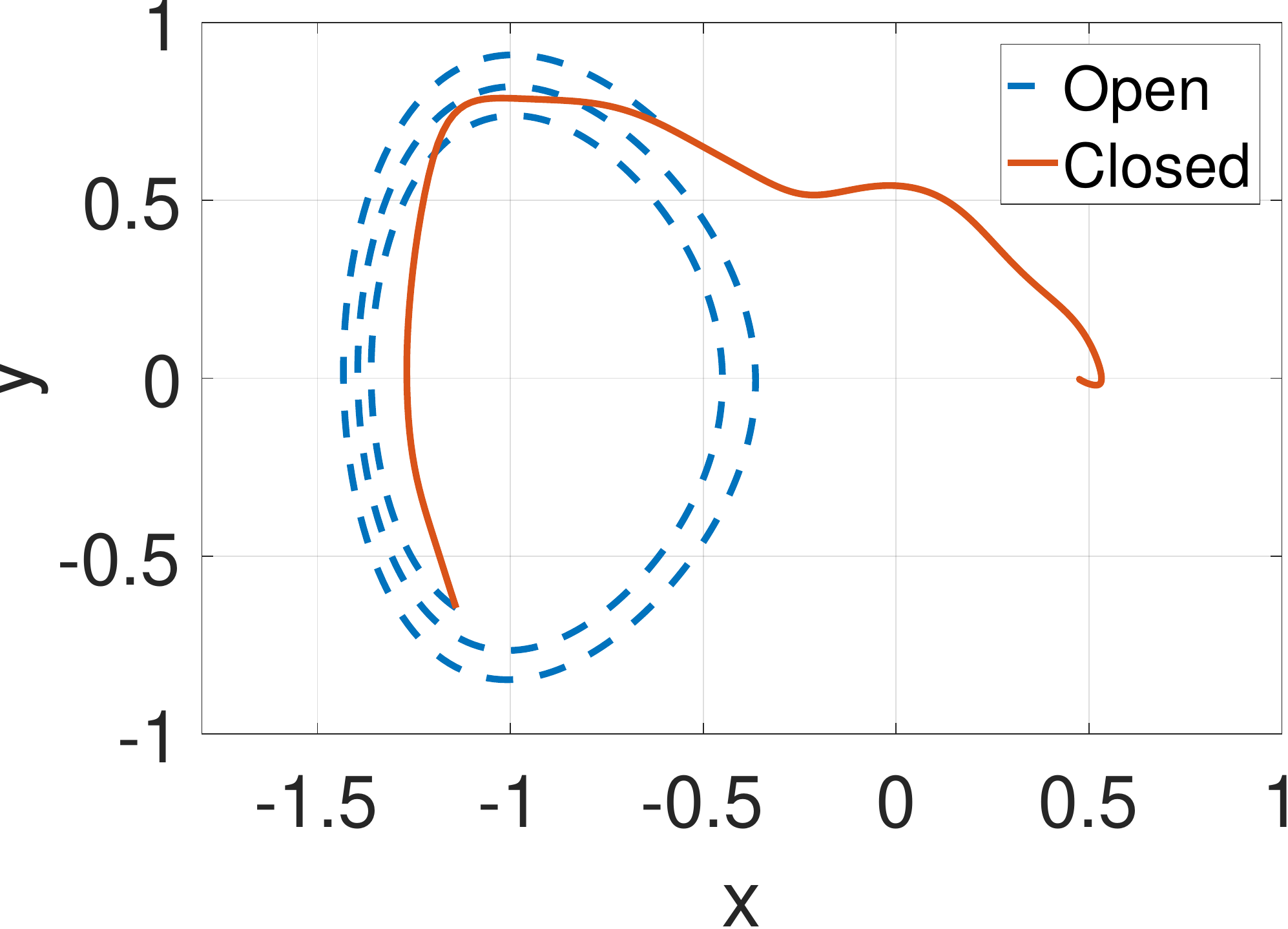}}\hspace{1em}%
 \subcaptionbox{\label{duffos_control_cost154}}{\includegraphics[width=1.57in,height=1.2in]{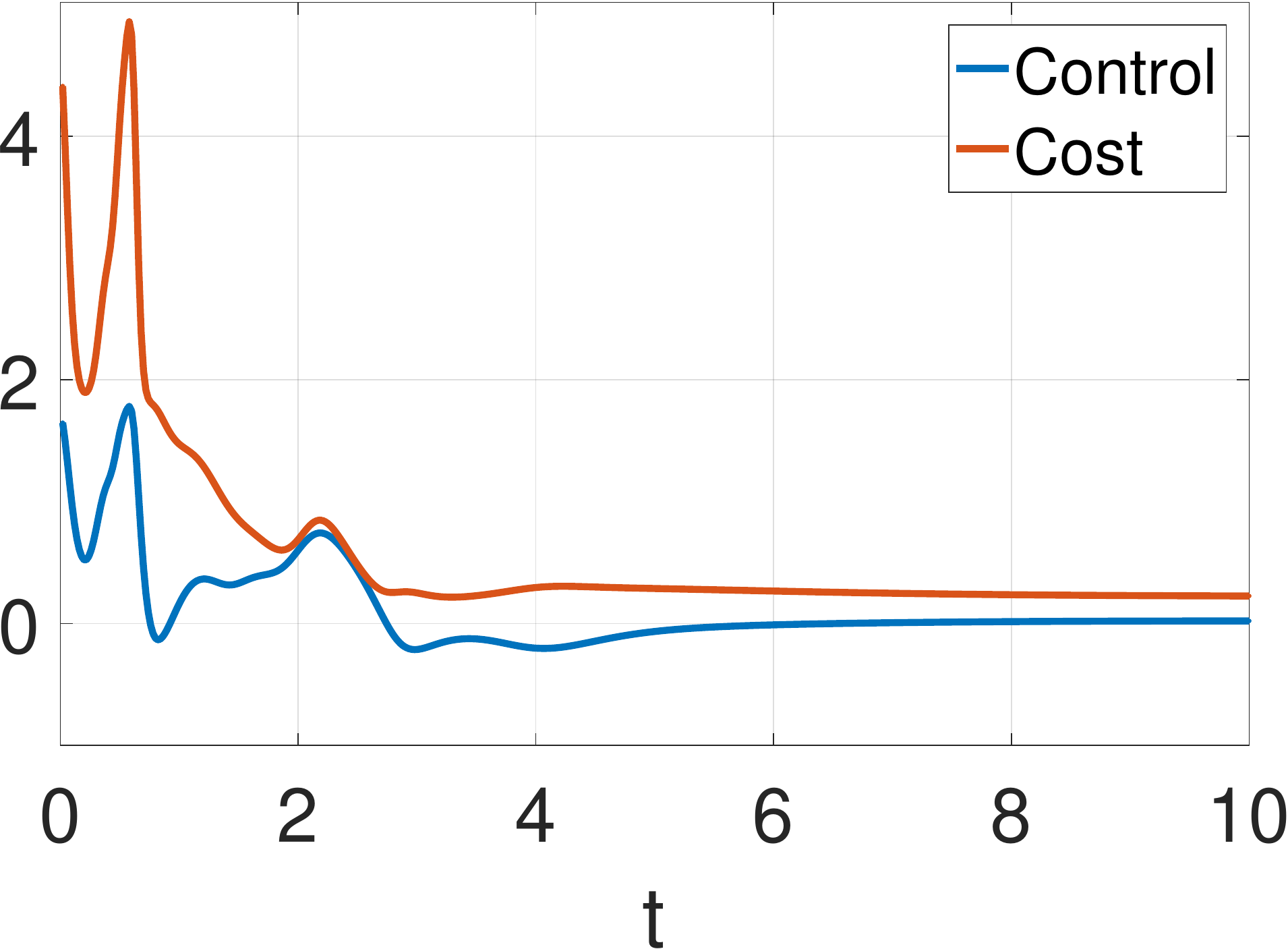}}
\caption{Basin Hopping Double Well: a) Open-loop and closed-loop trajectories; b) Optimal cost and control inputs.}
\label{Fig1_doublewell}
\end{center}
\end{figure}
%
% With Legend on state trajectories
\begin{figure}[]
\begin{center}
 \centering
  \subcaptionbox{\label{bhdwell_traj655}}{\includegraphics[width=1.57in,height=1.22in]{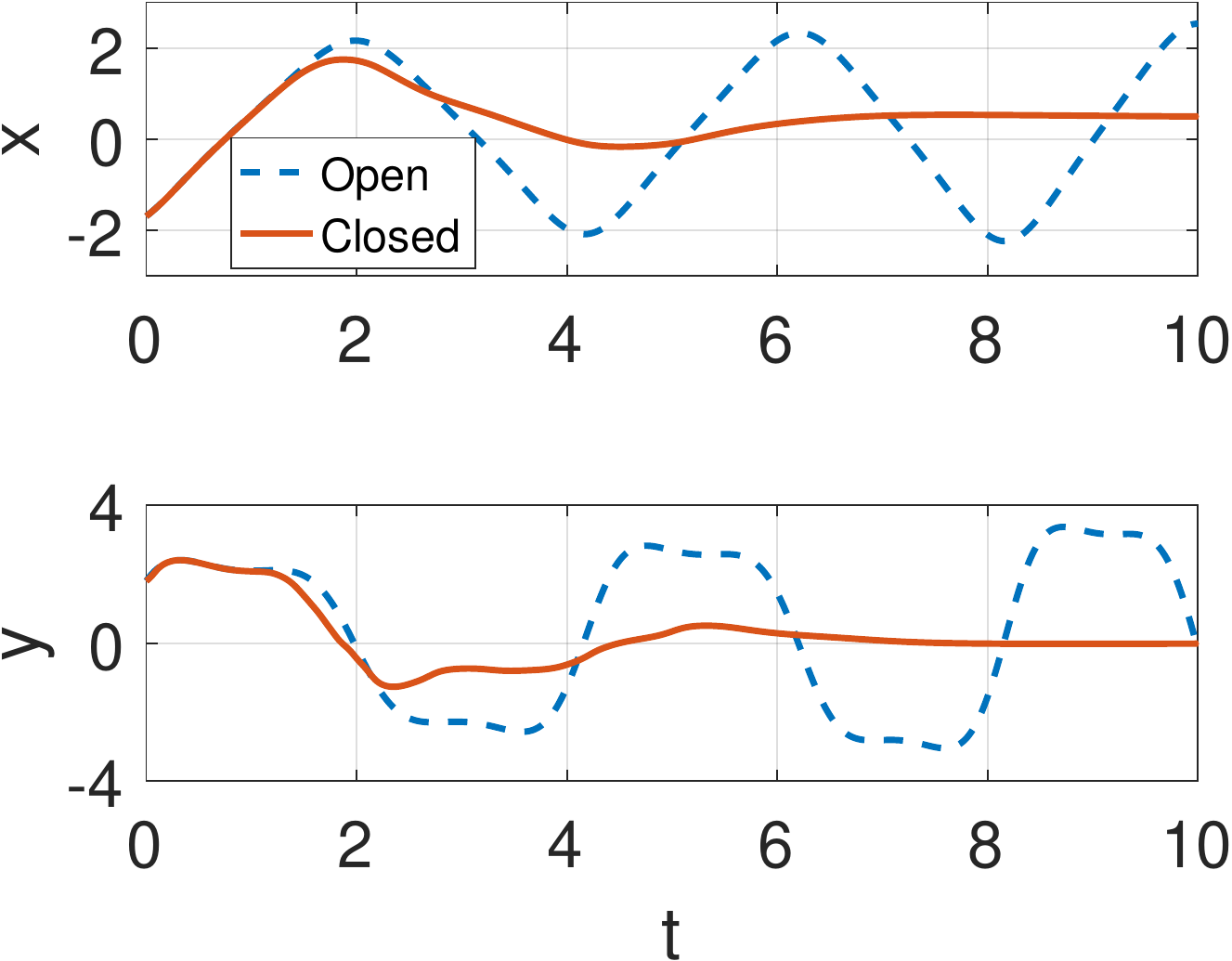}}\hspace{1em}%
 \subcaptionbox{\label{bhdwell_control_cost655}}{\includegraphics[width=1.57in,height=1.2in]{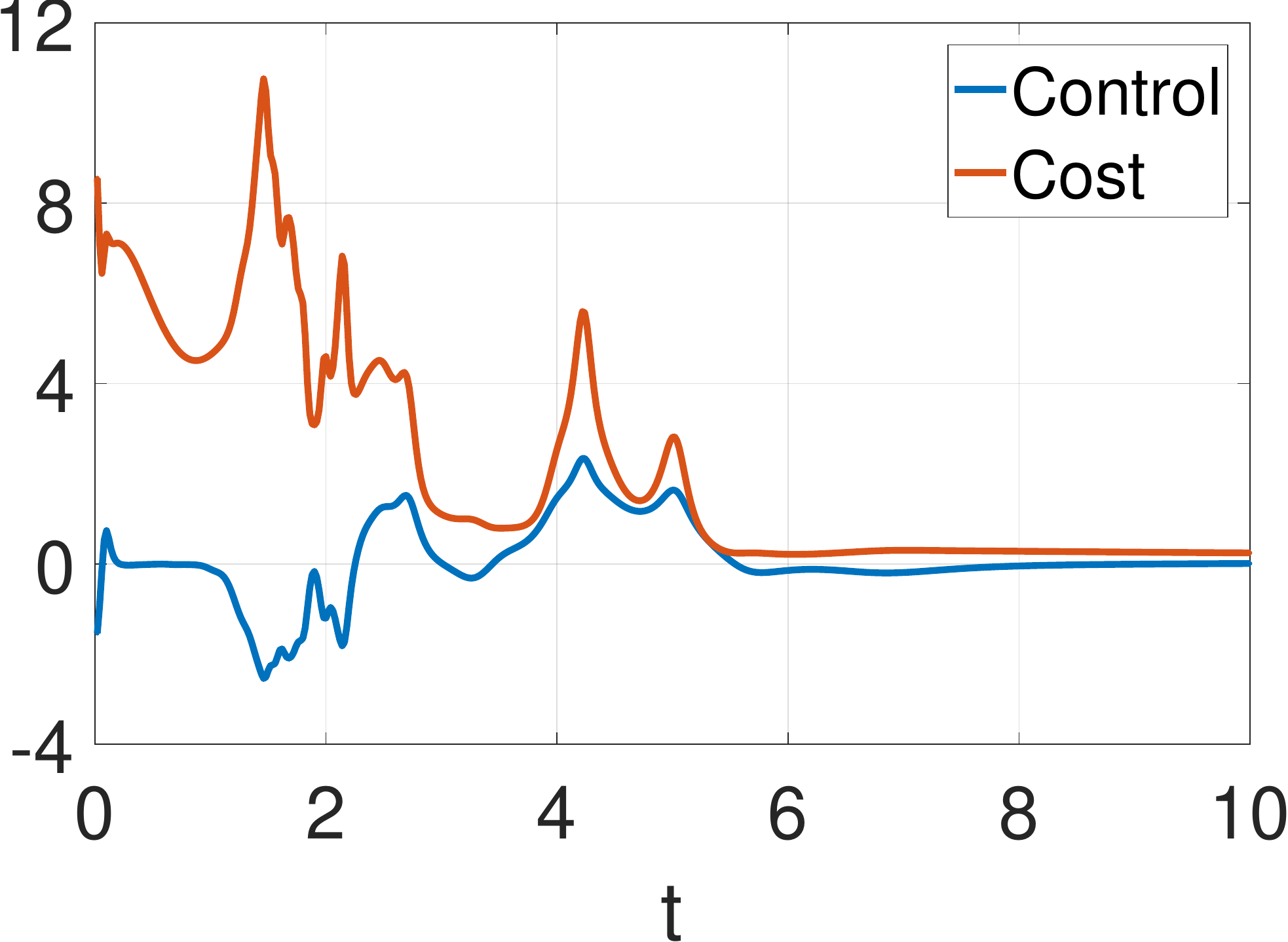}}
\caption{Basin Hopping Double Well: a) Open loop and closed loop trajectories; b) Optimal cost and control inputs.}
\label{Fig2_doublewell}
\end{center}
\end{figure}

\underline{\it Standard Map}
\begin{eqnarray}
x_{n+1}  &=&  x_n + y_n + K u \sin2\pi x_n \;\;(\mod 1)\nonumber \\
y_{n+1}  &=&  y_n + K u \sin 2\pi x_n \;\;
\label{system_standard_map}
\end{eqnarray}

Standard Map is one of the classical example of system exhibiting complex dynamics. The states of the standard map are canonical action-angle coordinates and they arise as a discretization of $1\frac{1}{2}$ degree of freedom Hamiltonian system. Control of standard maps are studied in \cite{vaidya2004controllability}.  For the uncontrolled standard map the entire state space $(x,y)\in [0,1]\times [0,1]$ is foliated with periodic and quasi periodic motion. The control objective is to stabilize the period 2 orbit located at $(0.25,0.5)$ and $(0.75,0.5)$. The parameter value of $K$ is chosen to be equal to $0.25$. For finite dimensional approximation, we used $200$ Gaussian radial basis  functions with $\sigma=0.02$. Control is quantized to ${\cal U}=[0.5:0.02:0.5]$. Open loop and closed loop control trajectories for the stabilization of period two orbit is shown in Fig. \ref{Fig1_standardmap}(a) with corresponding optimal control value is shown in Fig. \ref{Fig1_standardmap}(b). The stabilization of period two orbit is also evident from Fig. \ref{Fig3_standardmap}(a), where along the $x$ direction the system trajectory toggle between two points $x=0.25$ and $x=0.75$ and along $y$ axis the trajectory stabilize to $y=0.5$.

\begin{figure}[]
\begin{center}
 \centering
  \subcaptionbox{\label{stdmap_traj482_grid}}{\includegraphics[width=1.57in,height=1.22in]{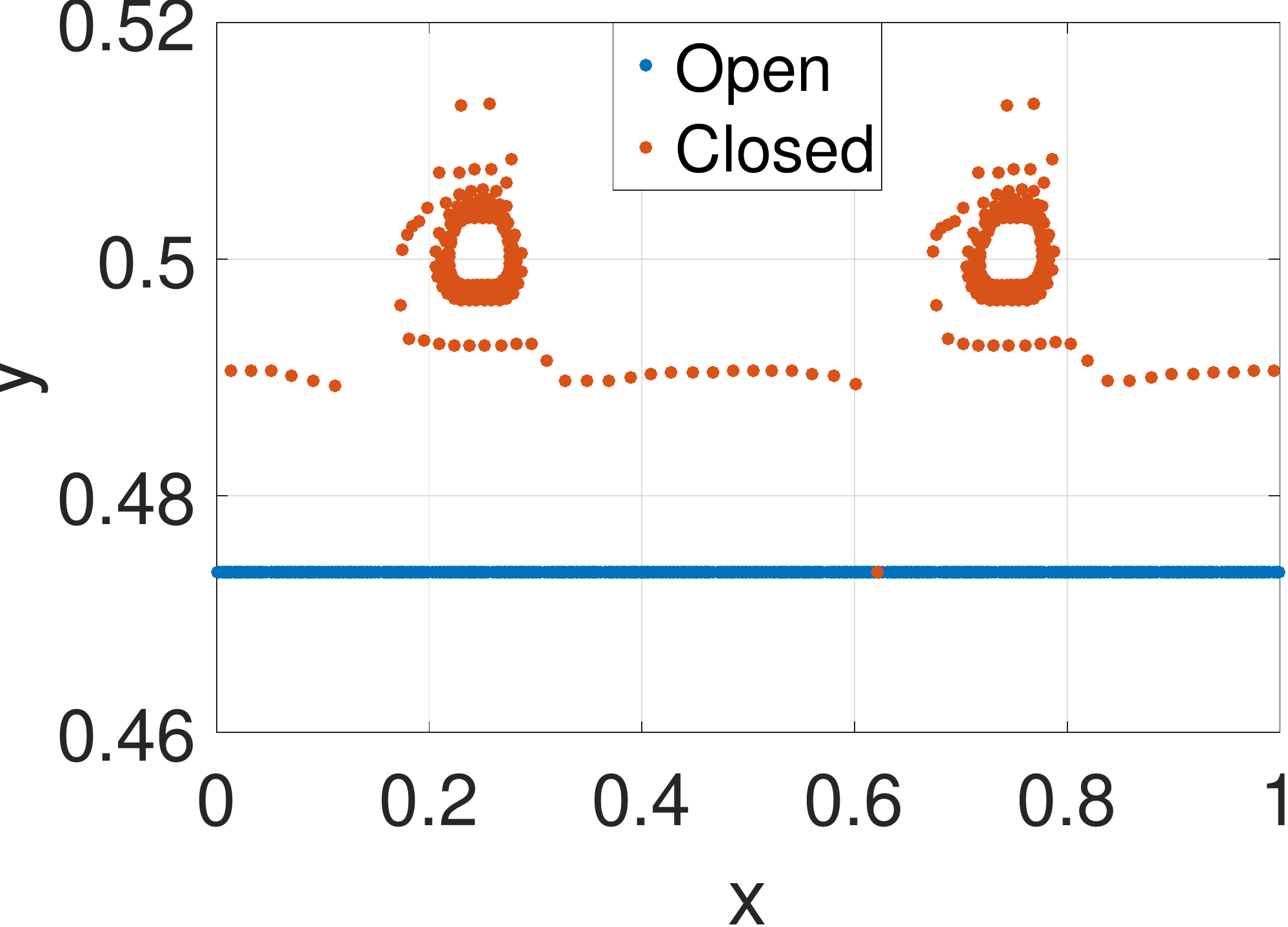}}\hspace{1em}%
 \subcaptionbox{\label{stdmap_control482_grid}}{\includegraphics[width=1.57in,height=1.22in]{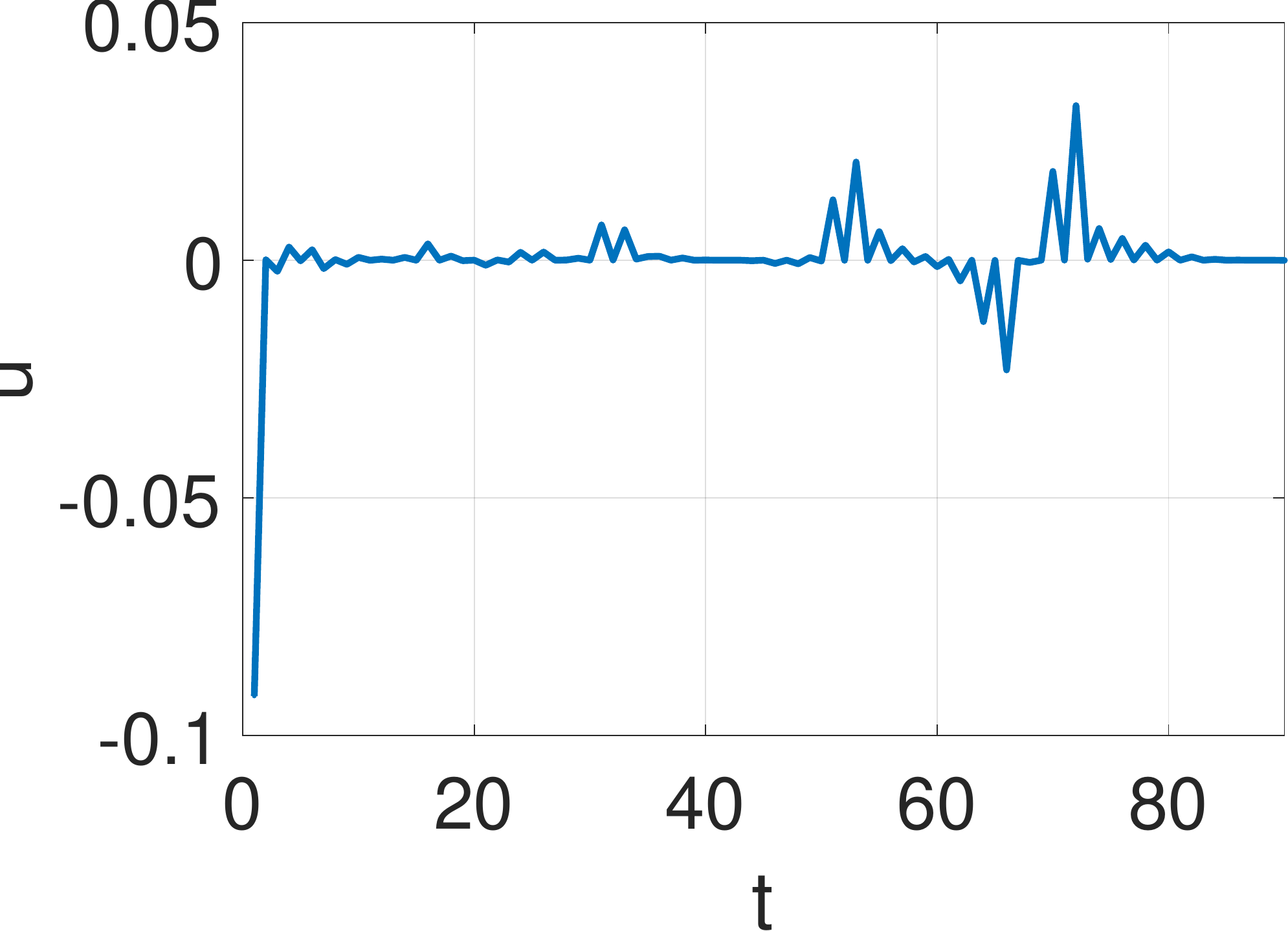}}
\caption{Period 2-orbit stabilization for standard map: a) Open-loop and closed loop trajectories; b) Optimal control value.}
\label{Fig1_standardmap}
\end{center}
\end{figure}

%\begin{figure}[H]
%\begin{center}
% \centering
 % \subcaptionbox{Sample Trajectories\label{stdmap_traj4405}}{\includegraphics[width=1.55in,height=1.2in]{stdMap_traj4405_grid}}\hspace{1em}%
 %\subcaptionbox{Control values\label{stdmap_control4405}}{\includegraphics[width=1.55in,height=1.2in]{stdMap_control4405_grid}}
%\caption{Standard Map trajectory 4}
%\label{Fig2_standardmap}
%\end{center}
%\end{figure}

\begin{figure}[]
\begin{center}
 \centering
  \subcaptionbox{\label{stdmap_traj_states1863_grid}}{\includegraphics[width=1.57in,height=1.22in]{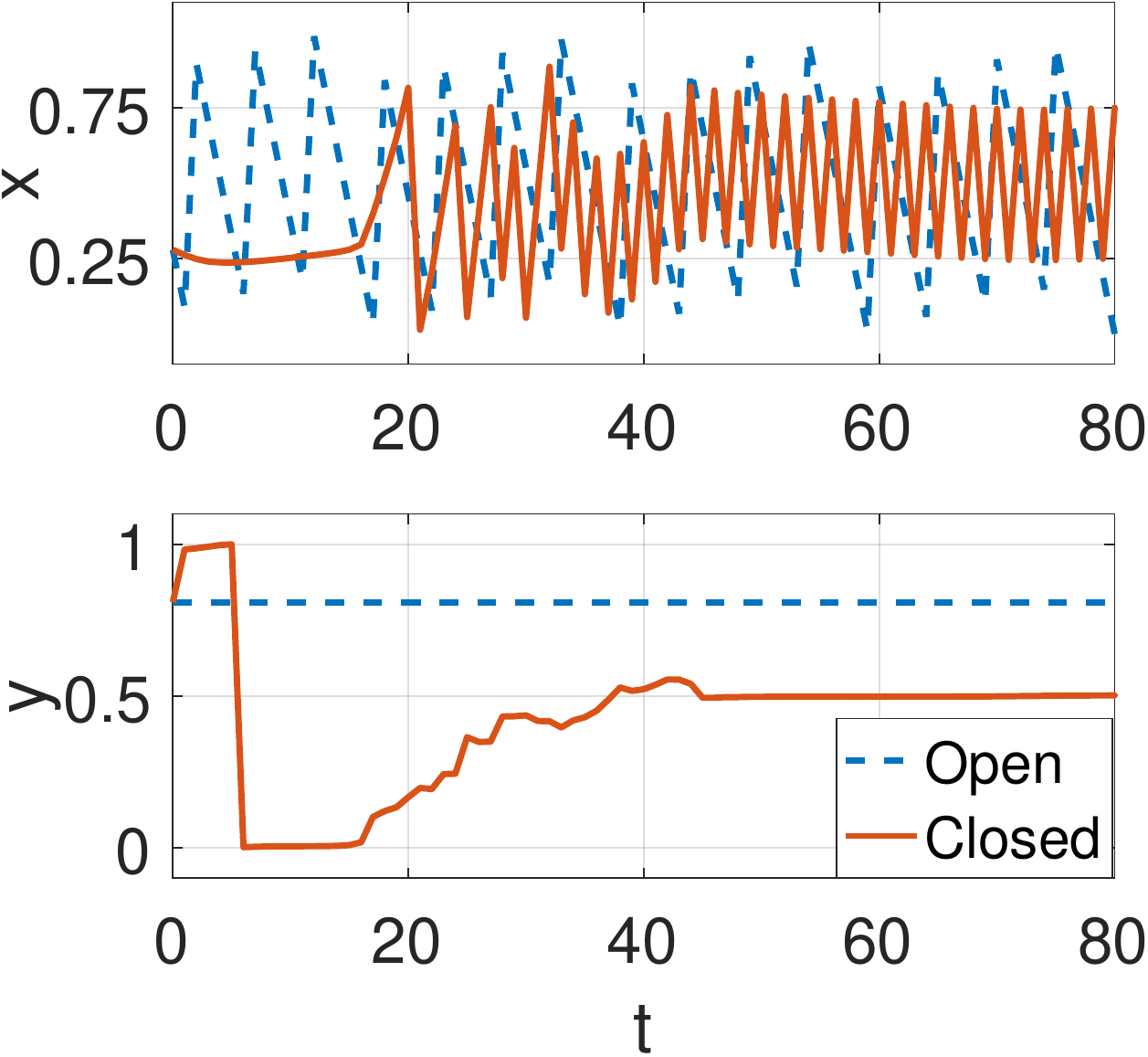}}\hspace{1em}%
 \subcaptionbox{\label{stdmap_control482_grid}}{\includegraphics[width=1.57in,height=1.22in]{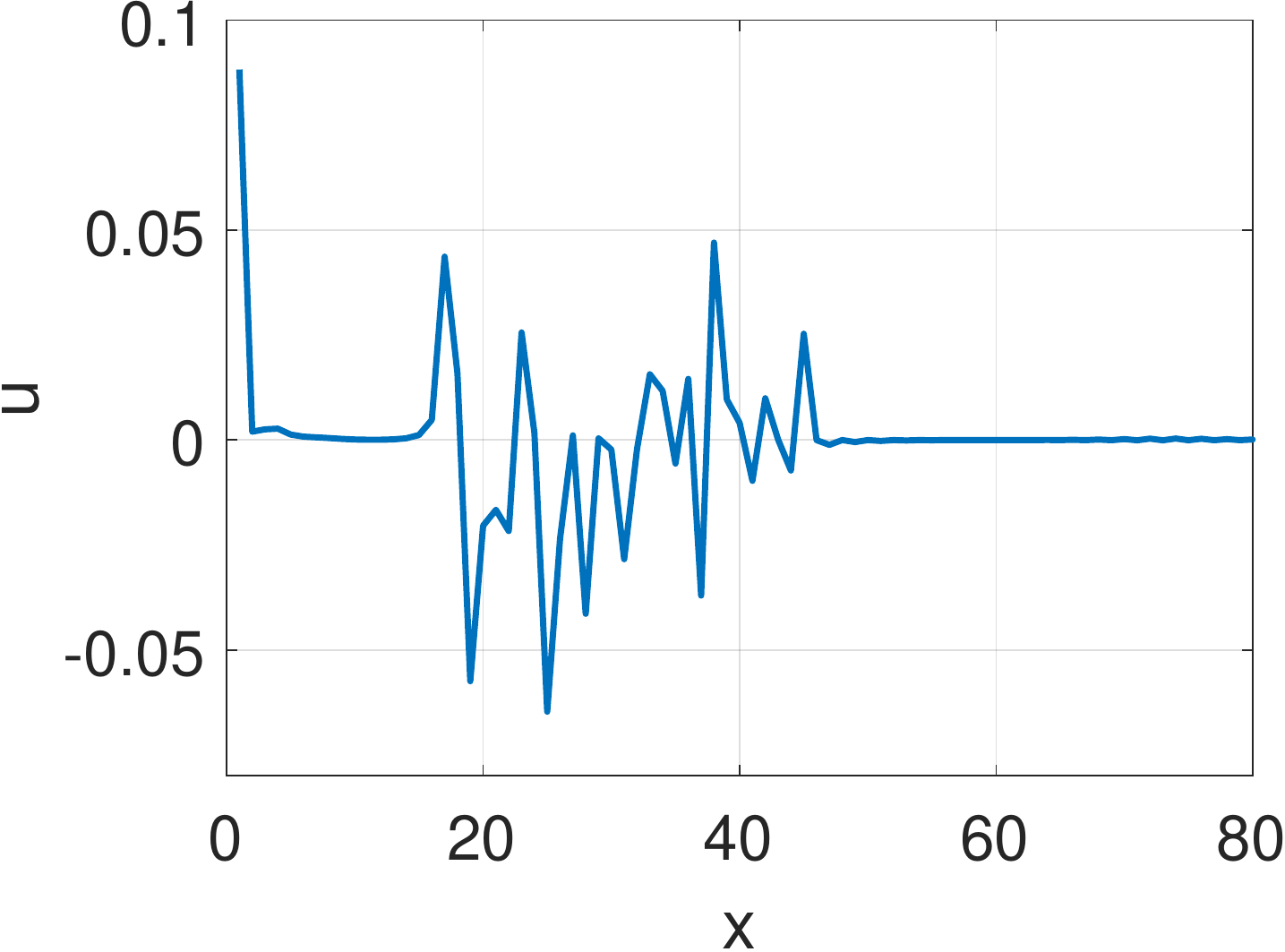}}
\caption{Period 2-orbit stabilization for standard map: a) Open-loop and closed loop trajectories; b) Optimal control value.}
\label{Fig3_standardmap}
\end{center}
\end{figure}

\section{Conclusions}
The main contribution of this paper is to utilize time-series data trajectory from a nonlinear system in order to provide {\it linear programming}-based approach for optimal stabilization of an attractor set. The proposed method relies on a linear transfer operator theoretic framework for a linear representation of a nonlinear system and the design of optimal stabilizing feedback controller. We use Naturally Structured Dynamic Mode Decomposition (NSDMD) algorithm for the finite-dimensional approximation of the transfer Koopman and then P-F operator from time series data.  The finite-dimensional approximation of the P-F operator is employed for the optimal stabilization of equilibrium point and periodic orbit using the Lyapunov measure-based optimal stabilization algorithm developed in \cite{raghunathan2014optimal}.

\bibliographystyle{ieeetr}
\bibliography{ref}

\begin{thebibliography}{10}

\bibitem{VaidyaMehtaTAC}
U.~Vaidya and P.~G. Mehta, ``Lyapunov measure for almost everywhere
  stability,'' {\em IEEE Transactions on Automatic Control}, vol.~53,
  pp.~307--323, 2008.

\bibitem{Vaidya_CLM}
U.~Vaidya, P.~Mehta, and U.~Shanbhag, ``Nonlinear stabilization via control
  {L}yapunov measure,'' {\em IEEE Transactions on Automatic Control}, vol.~55,
  pp.~1314--1328, 2010.

\bibitem{raghunathan2014optimal}
A.~Raghunathan and U.~Vaidya, ``Optimal stabilization using lyapunov
  measures,'' {\em IEEE Transactions on Automatic Control}, vol.~59, no.~5,
  pp.~1316--1321, 2014.

\bibitem{Lasota}
A.~Lasota and M.~C. Mackey, {\em Chaos, Fractals, and Noise: Stochastic Aspects
  of Dynamics}.
\newblock New York: Springer-Verlag, 1994.

\bibitem{Dellnitz00}
M.~Dellnitz and O.~Junge, {\em Set oriented numerical methods for dynamical
  systems}, pp.~221--264.
\newblock World Scientific, 2002.

\bibitem{Meic_model_reduction}
I.~Mezi\'{c}, ``Spectral properties of dynamical systems, model reductions and
  decompositions,'' {\em Nonlinear Dynamics}, 2005.

\bibitem{das2017transfer}
A.~K. Das, A.~U. Raghunathan, and U.~Vaidya, ``Transfer operator-based approach
  for optimal stabilization of stochastic systems,'' in {\em American Control
  Conference (ACC), 2017}, pp.~1759--1764, IEEE, 2017.

\bibitem{DMD_schmitt}
P.~J. Schmid, ``Dynamic mode decomposition of numerical and experimental
  data,'' {\em Journal of Fluid Mechanics}, vol.~656, pp.~5--28, 2010.

\bibitem{brunton2016discovering}
S.~L. Brunton, J.~L. Proctor, and J.~N. Kutz, ``Discovering governing equations
  from data by sparse identification of nonlinear dynamical systems,'' {\em
  Proceedings of the National Academy of Sciences}, vol.~113, no.~15,
  pp.~3932--3937, 2016.

\bibitem{kutz2016dynamic}
J.~N. Kutz, S.~L. Brunton, B.~W. Brunton, and J.~L. Proctor, {\em Dynamic Mode
  Decomposition: Data-Driven Modeling of Complex Systems}.
\newblock SIAM, 2016.

\bibitem{rowley2009spectral}
C.~W. Rowley, I.~Mezi{\'c}, S.~Bagheri, P.~Schlatter, and D.~S. Henningson,
  ``Spectral analysis of nonlinear flows,'' {\em Journal of fluid mechanics},
  vol.~641, pp.~115--127, 2009.

\bibitem{EDMD_williams}
M.~O. Williams, I.~G. Kevrekidis, and C.~W. Rowley, ``A data--driven
  approximation of the koopman operator: Extending dynamic mode
  decomposition,'' {\em Journal of Nonlinear Science}, vol.~25, no.~6,
  pp.~1307--1346, 2015.

\bibitem{klus2015numerical}
S.~Klus, P.~Koltai, and C.~Sch{\"u}tte, ``On the numerical approximation of the
  perron-frobenius and koopman operator,'' {\em arXiv preprint
  arXiv:1512.05997}, 2015.

\bibitem{Umesh_NSDMD}
B.~Huang and U.~Vaidya, ``Data-driven approximation of transfer operators:
  Naturally structured dynamic mode decomposition,'' in {\em
  {https://arxiv.org/abs/1709.06203}}, 2016.

\bibitem{mezic_koopmanism}
M.~Budisic, R.~Mohr, and I.~Mezic, ``Applied koopmanism,'' {\em Chaos},
  vol.~22, pp.~047510--32, 2012.

\bibitem{susuki2011nonlinear}
Y.~Susuki and I.~Mezic, ``Nonlinear koopman modes and coherency identification
  of coupled swing dynamics,'' {\em IEEE Transactions on Power Systems},
  vol.~26, no.~4, pp.~1894--1904, 2011.

\bibitem{surana_observer}
A.~Surana and A.~Banaszuk, ``Linear observer synthesis for nonlinear systems
  using koopman operator framework,'' in {\em {Proceedings of IFAC Symposium on
  Nonlinear Control Systems}}, (Monterey, California), 2016.

\bibitem{MVCDC05}
P.~G. Mehta and U.~Vaidya, ``On stochastic analysis approaches for comparing
  dynamical systems,'' in {\em Proceeding of IEEE Conference on Decision and
  Control}, (Spain), pp.~8082--8087, 2005.

\bibitem{kaiser2017data}
E.~Kaiser, J.~N. Kutz, and S.~L. Brunton, ``Data-driven discovery of koopman
  eigenfunctions for control,'' {\em arXiv preprint arXiv:1707.01146}, 2017.

\bibitem{peitz2017koopman}
S.~Peitz and S.~Klus, ``Koopman operator-based model reduction for
  switched-system control of pdes,'' {\em arXiv preprint arXiv:1710.06759},
  2017.

\bibitem{Vaidya_converse}
U.~Vaidya, ``Converse theorem for almost everywhere stability using {L}yapunov
  measure,'' in {\em Proceedings of American Control Conference}, (New York,
  NY), 2007.

\bibitem{disintegration}
H.~Furstenberg, {\em Recurrence in Ergodic theory and Combinatorial Number
  Theory}.
\newblock Princeston, New Jersey: Princeston University Press, 1981.

\bibitem{arvind_ocp_online}
A.~Raghunathan and U.~Vaidya, ``Optimal stabilization using {L}yapunov
  measures,'' (http://www.ece.iastate.edu/{$\sim$}ugvaidya/publications.html),
  2012.

\bibitem{vaidya2004controllability}
U.~Vaidya and I.~Mezi{\'c}, ``Controllability for a class of area-preserving
  twist maps,'' {\em Physica D: Nonlinear Phenomena}, vol.~189, no.~3,
  pp.~234--246, 2004.

\end{thebibliography}
\end{document}